\documentclass[11pt]{article}
\usepackage{amsmath}
\usepackage{fullpage}
\usepackage{graphicx}

\newtheorem{theorem}{Theorem}[section]

\newtheorem{lemma}[theorem]{Lemma}

\newenvironment{proof}{\begin{trivlist}
\item[\hspace{\labelsep}{\bf\noindent Proof: }]}{\qedsymb\end{trivlist}}
\newcommand{\qed}{\hfill\rule{2mm}{2mm}}
\newcommand{\qedsymb}{\hfill{\rule{2mm}{2mm}}}

\newcommand{\remove}[1]{}
\def\N{{\mathbf N}}
\def\R{{\cal R}}
\def\p{{\mathbf{p}}}
\def\pm{{\mathbf{p_{\min}}}}
\def\C{\overline C}
\def\D{\overline D}
\def\paths{{\cal P}}

\begin{document}

\sloppy

\begin{titlepage}

\title{Bicriteria Optimization in Routing Games}

\author{
Costas Busch\\
Computer Science Department\\
Louisiana State University\\
280 Coates Hall\\
Baton Rouge, LA 70803, USA\\
{\sf busch@csc.lsu.edu}
\and
Rajgopal Kannan\\
Computer Science Department\\
Louisiana State University\\
279 Coates Hall\\
Baton Rouge, LA 70803, USA\\
{\sf rkannan@csc.lsu.edu}
}

\date{}

\maketitle

\begin{abstract}
\noindent
Two important metrics for measuring the quality of routing paths
are the maximum edge congestion $C$ and maximum path length $D$.
Here, we study bicriteria in routing games
where each player $i$ selfishly selects a path that simultaneously minimizes
its maximum edge congestion $C_i$ and path length $D_i$.
We study the stability and price of anarchy of
two bicriteria games:
\begin{itemize}
\item
{\em Max games}, where the social cost is $\max(C,D)$
and the player cost is $\max(C_i, D_i)$.
We prove that max games are stable and convergent
under best-response dynamics, and that the price of anarchy
is bounded above by the maximum path length in the players' strategy sets.
We also show that this bound is tight in worst-case scenarios.

\item
{\em Sum games}, where the social cost is $C+D$
and the player cost is $C_i+D_i$.
For sum games, we first show the negative result that there are game instances
that have no Nash-equilibria.
Therefore,
we examine an approximate game called the {\em sum-bucket game}
that is always convergent (and therefore stable).
We show that the price of anarchy in sum-bucket games is bounded above by
$C^* \cdot D^* / (C^* + D^*)$ (with a poly-log factor),
where $C^*$ and $D^*$ are the optimal coordinated congestion and path length.
Thus, the sum-bucket game has typically
superior price of anarchy bounds than the max game.
In fact, when either $C^*$ or $D^*$ is small
(e.g. constant) the social cost of the Nash-equilibria is very close
to the coordinated optimal $C^* + D^*$ (within a poly-log factor).
We also show that the price of anarchy bound is tight for cases
where both $C^*$ and $D^*$ are large.
\end{itemize}
\end{abstract}


\vspace{2cm}
\noindent

\thispagestyle{empty}
\end{titlepage}

\section{Introduction}
\label{section:intro}
Routing is a fundamental task in communication networks.
Routing algorithms provide paths for packets that will be sent over the network.
There are two metrics that quantify the quality of the paths returned by a routing algorithm:
the congestion $C$, which is the maximum number of paths that use any edge in the network,
and the maximum path length $D$.
Assuming there is a packet for each path,
a lower bound on the
delivery time of the packets
is $\Omega(\max(C,D))$ (alternatively, $\Omega(C+D)$).
Actually, there exist packet scheduling algorithms that given the paths,
they deliver the packets along the paths
in time close to optimal $O(\max(C,D))$
(alternatively, $O(C+D)$)~\cite{CMS96,LMR94,LMR99,OR97,RT96}.

Motivated by the selfish behavior of entities in communication networks,
we study routing games where each packet's path is controlled independently
by a selfish player.
We model games with $N$ players,
where each player has to select a path
from a source to a destination node.
The objective of each player $i$ is to select a path
that simultaneously minimizes two parameters:
the congestion $C_i$, which is the maximum number of paths
that use any edge in player $i$'s path,
and the path length $D_i$.
We examine two kinds of games: {\em max games},
where the player's cost function is $\max(C_i, D_i)$,
and {\em sum games},
where the player's cost function is $C_i + D_i$.
In each of these games, the player's objective is to selfishly minimize its cost
in an uncoordinated manner.
From the player's point of view,
the minimization of the sum or max cost functions
are justified objectives,
since it is shown in \cite{scheideler-local} that player $i$'s
packet
can be delivered
in time
$\tilde O(C_i+D_i)$ (alternatively, $\tilde O(\max(C_i, D_i)$).

A natural problem is to determine the effect of the players' selfishness
on the welfare of the whole communication network.
In the max and sum games,
the welfare of the network is measured with the social cost functions $\max(C,D)$
and $C+D$, respectively.
The choice of these social cost functions is appropriate since
they determine the total time needed
to deliver
the packets represented by the players.
We examine the consequence of the
selfish behavior in Nash equilibria,
which are stable states of the game where no player can
unilaterally improve her situation.
The effect of selfishness is quantified with the
{\em price of anarchy} ($PoA$)~\cite{KP99,P01},
which expresses how much larger is the worst social cost in a Nash equilibrium
compared to the social cost in the optimal coordinated solution.
We study the existence of Nash equilibria and the price of anarchy for
max games and sum games,
where we find that these games produce different results with their own merits.

\subsection{Max Games}

First, we examine max games (the social cost is $\max(C+D)$).
We prove that every max game has at least one Nash equilibrium.
The equilibrium can be obtained by best response dynamics,
where a player greedily changes, whenever possible, the current path
to an alternative path with lower cost.
With best response moves
the game eventually converges to a Nash-equilibrium.
We show that the optimal
coordinated solution is a Nash-equilibrium too.
Thus, max games games have very good Nash equilibria.
This observation is quantified in terms of the
{\em price of stability} $(PoS)$~\cite{anshelevich2,anshelevich1}
which expresses how much larger is the best social cost in a Nash equilibrium
with respect to the social cost in the optimal coordinated solution.
Therefore in max games it holds that $PoS = 1$.

We then examine the worst Nash equilibria.
We bound the price of anarchy ($PoA$) in max games
with respect to the maximum allowable path length $L$ for the players in the network,
and the number of nodes $n$ in the graph:
$$PoA = O(L + \log n).$$
We prove that this bound is worst case optimal (within additive terms).
Specifically, we provide an example game in a ring network
where the optimal coordinated social cost is 1,
while there is a Nash equilibrium with cost $O(L) = O(n)$.

\subsection{Sum Games}

We continue with examining sum games (the social cost is $C+D$).
Intuitively, sum games have the potential to give better price of anarchy than
the max games because both parameters ($C$ and $D$) affect the choices at
all the time, even when one parameter is larger than the other.
For example,
the ring game that we mentioned above
has price of anarchy equal to $1$ in the sum game.
However, we prove a limitation of sum games:
not all sum games have Nash equilibria;
there exist instances of sum games with a small number of players
that do not have Nash equilibria at all.
This limitation directed us towards
exploring alternative games
which are stable (have Nash equilibria)
and have similar characteristics with the original sum games.

We found such a game variation that we call {\em sum-bucket game}.
In sum-bucket games the players are divided into $\log n$ classes, called {\em buckets},
according to the packet paths that they choose.
Bucket $k$ holds the paths of players
with length in range $[2^k, 2^{k-1})$.
Suppose that player $i$'s path is in bucket $k$.
The {\em normalized congestion} of player $i$, denoted $\C_i$,
is measured with respect to the paths that belong to bucket $k$.
The {\em normalized length} of player $i$'s path is $\D_i = 2^{k-1} -1$,
(which is a factor 2 approximation of the original length).
Player $i$'s cost function is $\C_i + \D_i$.
Thus, in sum-bucket games only players in the same bucket compete
with each other, while players in different buckets do not interfere.
The {\em normalized social cost} function is defined to be $\C + \D$,
where $\C$ is the maximum normalized congestion in any bucket,
and $\D$ is the maximum normalized depth of all paths.

We first show that sum-bucket games always have Nash equilibria,
which can be obtained with best response dynamics.
We then examine the quality of the Nash equilibria.
For every game there is a corresponding coordinated bucket routing problem.
We can bound the price of anarchy $PoA$
with respect to the normalized optimal congestion $\C^*$
and normalized path length $\D^*$ in the optimal coordinated solution.
We obtain:
$$
PoA = O \left ( \frac { \C^* \cdot \D^* } {\C^* + \D^*} \cdot \log^2 n \right )
$$
Therefore, when either of $\C^*$ or $\D^*$ is small (e.g. a constant),
the Nash equilibrium provides a very good approximation (within a poly-log factor from optimal)
to the uncoordinated routing problem.
In such scenarios, the price of selfishness is small.
However, when both $\C^*$ or $\D^*$ are simultaneously large,
the approximation becomes worse
(though still typically lower than the PoA of the max game since
the PoA is bounded by the smaller of $\C^*$ or $\D^*$, where $\D^* < 2 L$).
Nevertheless, even in these scenarios the PoA bound is tight
in certain games instances.

Sum-bucket games are interesting variations of sum games because
they are stable,
and they can be used to approximate
solutions for the $C+D$ social cost.
For any sum-bucket game, there is a
corresponding ``original'' coordinated routing problem
where the objective is to minimize the social cost $C+D$ without using buckets.
It holds that $\C \leq C \leq \C \cdot \log n$ and $\D \leq D \leq 2 \D$;
thus, $C+D = O(\C \cdot \log n + \D)$.
In other words, the normalized social cost can be used as an approximation
for the ``original'' social cost.
Let $PoA'$ denote how much larger is the worst equilibrium
of a sum-bucket game compared to the optimal solution of the
coordinated original problem (with respect to the social cost $C+D$).
It holds that $PoA' \leq PoA \cdot \log n$.
Consequently, the observations that we made above for the
$PoA$ in sum-bucket games apply also
with respect to the original routing problem.
For example,
when one of $C^*$ or $D^*$ is small (e.g. a constant),
then the Nash equilibrium of the sum-bucket game
provides a very good approximation (within a poly log factor from optimal)
to the coordinated original routing problem.

\subsection{Related Work}

Routing games (on congestion) were introduced and
studied in~\cite{monderer1,rosenthal1}.
The notion of price of anarchy was introduced in \cite{KP99}.
Since then, many routing game models have been studied
which are distinguished
by the topology of the network,
cost functions,
type of traffic (atomic or splittable),
nature of strategy sets,
and kind of equilibria (pure or mixed).
Specifically,
pure equilibria with atomic flow have been studied in
\cite{BM06,CK05,libman1,rosenthal1,STZ04} (our work fits into this category),
and with splittable flow in
\cite{roughgarden1,roughgarden2,roughgarden3,roughgarden5}.
Mixed equilibria with atomic flow have been studied in
\cite{CKV02,czumaj1,FKP02,GLMMb04,GLMM04,GLMMR04,KMS02,KP99,LMMR04,MS01,P01},
and with splittable flow in
\cite{correa1,FKS02}.

To our knowledge there is no previous work that considers
routing games that optimize two criteria simultaneously.
Most of the work in the literature uses a single cost metric
which is related to the congestion.
A common metric for the player cost
is the sum of the congestions on all the edges of the player's path
(we denote this kind of player cost as $pc'$)
and the respective social cost is the cost of the worst player's path
(we denote this social cost as $SC'$)
\cite{CK05,GLMMR04,roughgarden2,roughgarden3,roughgarden5,STZ04}.
However, as we discussed before,
in packet scheduling algorithms,
the $pc'$ or $SC'$ do not govern
the packet delays;
$\max(C_i, D_i)$ or $C_i + D_i$ govern the packet delay.

Other combinations of player costs and social costs
have been studied in the literature:
player cost $pc'$ and social cost $C$ has been studied in
\cite{CK05,correa1,CKV02,czumaj1,FKP02,FKS02,GLMMb04,KMS02,KP99,MS01,P01,roughgarden1};
player cost $C_i$ and social cost $C$ has been studied in \cite{BM06};
other variations have been studied in \cite{GLMM04,libman1,LMMR04,rosenthal1}.
The vast majority of the work on routing games has been performed
for parallel link networks,
with only a few exceptions on general network topologies
\cite{BM06, CK05,correa1,roughgarden1}.

Our work is closer to \cite{BM06}.
We extended some results presented in \cite{BM06}
to apply to bicriteria,
instead of the single criterium of congestion, player cost $C_i$ and social cost $C$,
that was used in~\cite{BM06}.
Specifically, the particular techniques that we use
to prove existence of Nash equilibria with best response dynamics,
and also to prove upper bounds on the price of anarchy,
were originally introduced in~\cite{BM06}.
Here, we modified and extended appropriately these techniques
in a non-trivial way to apply to our new cost functions.


\subsection*{Outline of Paper}
We proceed as follows.
In Section \ref{section:definitions}
we give basic definitions.
We study max games in Section \ref{section:max-games}
and sum games in Section \ref{section:sum-games}.
We finish with our conclusions in Section \ref{section:conclusions}.
Due to space limitations some proofs have moved to the appendix.
\section{Definitions}
\label{section:definitions}

An instance of a \emph{routing game} is a tuple
$\R = (\N,G,\paths)$,
where
$\N=\{1,2,\ldots,N\}$ are the players,
$G = (V,E)$ is a graph with nodes $V$ and edges $E$,
and the graph has paths $\paths = \bigcup_{i \in \N} \paths_i$,
where $\paths_i$ is a collection of
available paths in $G$ for player $i$.
Each path in
$\paths_i$ is a path in $G$ that
has the same source
$u_i\in V$ and destination
$v_i\in V$;
each path in $\paths_i$ is a
{\em pure strategy} available to player $i$.
A {\em pure strategy profile}
$\p=[p_1,p_2,\cdots,p_N]$ is a collection of pure
strategies (paths),
one for each player, where $p_i\in\paths_i$.
We refer to a pure strategy profile as a \emph{routing}.
On a finite network, a routing game is necessarily
a finite game.

For any routing $\p$ and any edge $e\in E$, the
\emph{edge-congestion}  $C_e(\p)$ is the number of paths in
$\p$ that use edge $e$.
For any path $p$, the \emph{path-congestion}
$C_p(\p)$ is the maximum edge congestion over all
edges in $p$, $C_p(\p)=\max_{e\in p}C_e(\p)$.
We will use the notation $C_i(\p) = C_{p_i}(\p)$,
for any user $i$.
The \emph{network congestion}
is the maximum edge-congestion over all edges in $E$,
that is,
$C(\p)=\max_{e\in E}C_e(\p)$.
We denote the length (number of edges) of any path $p$ as $|p|$.
For any user $i$,
we will also use the notation $D_{p_i}(\p)$ or $D_i(\p)$
to denote the length $|p_i|$.
The {\em longest path length} in $\paths$
is denoted $L(\paths) =\max_{p\in \paths}|p|$.
We will denote by $D(\p)$ the maximum path length
in routing $\p$, that is $D(\p) = \max_{p \in \p} |p|$.
When the context is clear, we will
drop the dependence on $\p$ and $\R$
and use the notation $C_e,C_p,C_i,C,L,D_p,D_i,D$.

For game $R$ and routing $\p$,
the \emph{social cost} (or {\em global cost})
is a function of routing $\p$, and it is
denoted $SC(\p)$.
The \emph{player or local cost} is also a function on $\p$
denoted $pc_i(\p)$.
We use the standard notation
$\p_{-i}$ to refer to the collection of paths
$\{p_1,\cdots,p_{i-1},p_{i+1},\cdots,p_N\}$, and
$(p_i;\p_{-i})$ as an alternative notation for $\p$ which
emphasizes the dependence on $p_i$.
Player $i$
is \emph{locally optimal} in routing $\p$ if
$pc_i(\p) \leq pc_i(p_i';\p_{-i})$ for
all paths $p_i'\in\paths_i$.
A routing $\p$ is in a Nash Equilibrium
(we say $\p$ is a \emph{Nash-routing}) if
every player is locally optimal.
Nash-routings quantify the notion of
a stable selfish outcome.
A routing $\p^*$ is an optimal pure strategy profile
if it has minimum
attainable social cost: for any other pure strategy profile
$\p$, $SC(\p^*)\le SC(\p)$.

We quantify the quality of the Nash-routings by the
\emph{price of anarchy} ($PoA$)
(sometimes referred to as the coordination ratio)
and the
\emph{price of stability} ($PoS$).
Let $\bf P$ denote the set of distinct
Nash-routings, and let
$SC^*$ denote the
social cost of an optimal routing $\p^*$.
Then,
%
\begin{equation*}
PoS
= \inf\limits_{\p\in ~{\bf P}} \frac{SC(\p)}{SC^*},
\qquad
PoA
=\sup\limits_{\p\in ~{\bf P}} \frac{SC(\p)}{SC^*}.
\end{equation*}

\section{Max Games}
\label{section:max-games}

Let $\R = (\N,G,\paths)$
a routing game
such that for any routing $\p$
the social cost function is $SC(\p) = \max(C(\p),D(\p))$,
and the player cost function $pc_i(\p) = \max(C_i(\p), D_i(\p))$.
We refer to such routing games as {\em max games}.
First, we show that max games have Nash-routings
and the price of stability is 1.
Then, we bound the price of anarchy.

\subsection{Existence of Nash-routings in Max Games}
\label{section:max-stability}

We show that max games have Nash-routings.
We prove this result by first giving a totaly order for the routings
using a form of lexicographic ordering.
Then we show that any greedy move of a player can only
give a new routing with smaller order.
Thus, the greedy moves will converge either to the smallest routing
or to a routing where no player can improve further.
In either case, a Nash-routing will be reached.

Let $\R = (\N,G,\paths)$
be a max routing game.
Let $r = \max(N,L)$.
For any routing $\p$
we define
the \emph{routing vector}
$M(\p)=[m_1(\p), \ldots, m_r(\p)]$,
where $m_i(\p) = a_i(\p) + b_i(\p)$,
and $a_i(\p)$ is the number of paths with congestion $i$,
and $b_i(\p)$ is the number of paths with length $i$.
Note that if $SC(\p) = k$ then $m_{k} \neq 0$ and $m_k' = 0$ for all $k' > k$.

We define a total order on the routings as follows.
Let $\p$ and $\p'$ be two routings, with
$M(\p)=[m_1,\ldots,m_r]$, and
$M(\p')=[m'_1,\ldots,m'_r]$.
We say that $M(p) = M(p')$ if
$m_i = m'_i$ for all $1 \leq i \leq r$.
We say that $M(\p) < M(\p')$ if there is
a $j$, $1 \leq j \leq r$,
such that $m_k = m'_k $ for all $k > j$,
and $m_j < m'_j$.
We order the $p$ and $p'$ according to the
order of their respective vectors,
that is $p \leq p'$ if and only if $M(p) \leq M(p')$.
Note that for any two $p$ and $p'$ it either holds that $p = p'$
or $p < p'$.
That is, the routings are totally ordered.

Consider an arbitrary routing $\p$.
If $\p$ is not a Nash-routing,
there is at least one user $i$ which is not locally optimal.
Then a {\em greedy move} is available to player $i$
in which the player can obtain lower cost by changing the path
from $p_i$ to some other path $p'_i$ with lower cost.
In other words, the greedy move takes the original routing
$\p = (p_i;\p_{-i})$ to a routing $\p' = (p_i';\p'_{-i})$
with improved player cost $pc_i(\p') < pc_i(\p)$,
such that $p_i$ is replaced by $p_i'$
and the remaining paths stay the same
($\p_{-i} = \p'_{-i}$).
We show now that any greedy move gives a smaller order routing:

\begin{lemma}
\label{theorem:max-greedy}
If a greedy move by any player takes a routing
$\p$ to a new routing
$\p'$, then
$\p' < \p$.
\end{lemma}

\begin{proof}
Let $pc_i(\p) = \max(C_i(\p), D_i(\p)) = k_1^{\max}$,
and $\min(C_i(\p), D_i(\p)) = k_1^{\min}$ (clearly, $k_1^{\max} \geq k_1^{\min}$).
Let also $pc_i(\p') = \max(C_i(\p'), D_i(\p')) = k_2^{\max}$,
and $\min(C_i(\p'), D_i(\p')) = k_2^{\min}$ (clearly, $k_2^{\max} \geq k_2^{\min}$).
Since player $i$ can decrease its cost in $\p'$, $k_2^{\max} < k_1^{\max}$.
Consider now the vectors $M(\p) = [m_1,\ldots,m_r]$
and $M(\p') = [m'_1,\ldots,m'_r]$.
These two vectors are the same except possibly for entries
$k_1^{\max}$, $k_1^{\min}$, $k_2^{\max}$, $k_2^{\min}$,
which correspond to the positions
that are affected by paths $p_i$ and $p'_i$.
It holds that $m_{k_1^{\max}} > m'_{k_1^{\max}}$,
since when the path switches to $\p_i'$,
$m_{k_1^{\max}} = a_{k_1^{\max}} + b_{k_1^{\max}}$
decreases by at least one because either $a_{k_1^{\max}}$ decreases by one
(if the new path has lower congestion)
or $b_{k_1^{\max}}$ decreases by one (if the new path has lower length).
Since $k_2^{\max} < k_1^{\max}$,
$M(\p) > M(\p')$ implying that $\p > \p'$.
\end{proof}

Since there are only a finite number of routings,
Lemma \ref{theorem:max-greedy} implies that starting from arbitrary
initial state,
every best response dynamic converges in a finite time
to a Nash-routing, where every player is locally optimal.
Since the routings are totally ordered,
there is a routing $\pm$ which is the minimum,
that is, for all routings $\p$, $\pm \leq \p$.
Clearly, the minimum routing is also a Nash-routing.
The minimum routing $\pm$
achieves also optimal social cost,
since if there was another routing $\p'$ with lower social cost,
then it can be easily shown that $\p' < \pm$,
which is contradiction.
Thus, the price of stability is $1$.
Therefore, we have the following result:

\begin{theorem}[Stability of max games]
\label{theorem:max-stability}
For any max game $\R$,
every best response dynamic converges to a Nash-routing,
and the price of stability is $PoS = 1$.
\end{theorem}

\subsection{Price of Anarchy in Max Games}
\label{section:expansion}

We bound the price of anarchy in max games.
Consider a max routing game $\R = (\N,G,\paths)$,
where $G$ has $n$ nodes.
Theorem \ref{theorem:max-stability} implies that there is at least one Nash-routing.
Consider a Nash-routing $\p$.
Denote $C = C(\p)$ and $D = D(\p)$.
Let $\p^*$ be the optimum (coordinated) routing
with minimum social cost.
Denote $C^* = C(\p^*)$ and $D^* = D(\p^*)$.
Note that each payer $i \in \N$
has a path $p_i \in \p$ and a corresponding ``optimal'' path $p^*_i \in \p^*$
from the player's source to the destination.

For each edge $e \in G$, denote $\Pi_e(\p)$ the set of players
whose paths in routing $\p$ use edge $e$.
We define $H$ to be a set that contains all edges $e \in G$ with congestion $C_e(\p) \geq D + 2$.
Consider an edge $e \in H$.
Let $i \in \Pi_e(\p)$ be a player whose path $p_i$ in routing $\p$
uses edge $e$.
We define $f(e,i)$ to be a set that contains all edges $e' \in p^*_i$ with $C_{e'}(\p) \geq C_e(\p) - 1$.
It holds that $|f(e,i)| \geq 1$,
since in routing $\p$ player $i$
prefers path $p_i$ instead of $p^*_i$
because there is at least one edge $e' \in p^*_i$ with $C_{e'}(\p) \geq C_e(\p) - 1 > D$.
Let $f(e) = \cup_{i \in \Pi_e(\p)} f(e,i)$.
For any set of edges $X \subset H$,
we define $f(X) = \bigcup_{e \in X} f(e)$.

\begin{lemma}
\label{theorem:max-expansion}
Let $Z$ be the set that contains all
edges $e$ with congestion $C_e(\p) \geq C - 2\log n$.
If $C \geq D + 2\lg n + 2$,
then there is a set of edges $X \subseteq Z$
such that $|f(X)| \leq 2|X|$.
\end{lemma}

\begin{proof}
We recursively construct a sets of edges $E_0, \ldots, E_{2\lg n}$,
such that $E_i = E_{i-1} \cup f(E_{i-1})$,
and set $E_0$ contains
all the edges $e$ with congestion $C_e(\p) = C$.
From the construction of those sets
it holds that for any $e \in E_j$, $C_{e}(\p) \geq C-j$,
where $0 \leq j \leq 2\lg n$.
Thus, $E_j \subseteq Z$, for all $0 \leq j \leq 2\lg n$.
(Note that $Z \subseteq H$.)

We can show that there is a $j$, $0 \leq j \leq 2\lg n$,
such that $|f(E_j)| \leq 2 |E_j|$.
Suppose for contradiction that such a $j$ does not exist.
Thus for all $j$, $0 \leq j \leq 2\lg n$,
it holds that $|f(E_j)| > 2|E_j|$.
In this case, it it straightforward to show that $|E_k| > 2|E_{k-1}|$,
for any $1 \leq k \leq 2\lg n$.
Since $|E_0| \geq 1$,
it holds that $|E_{2\lg n}| > 2^{2\lg n} = n^2$.
However, this is a contradiction, since
the number of edges in $G$ do not exceed $n^2$.
\end{proof}


\begin{lemma}
\label{theorem:max-upper-bound}
If $C \geq D + 2 \lg n + 2$, then $C < 2 L C^* + 2\lg n$.
\end{lemma}

\begin{proof}
From Lemma \ref{theorem:max-expansion},
there is a set of edges $X \subset Z$ with $|f(X)| \leq 2|X|$.
For each $e \in Z$ it holds that $C_e(\p) \geq C - 2 \lg n$.
Let $M = \sum_{e \in X} C_e(\p) \geq |X| (C - 2\lg n)$,
where $M$ denotes the total utilization of the edges
in $X$ by the paths of the players in $\Pi$.
Let $\Pi = \bigcup_{e \in X} \Pi_e(\p)$,
that is, $\Pi$ is the set of players which in routing $\p$
their paths use edges in $X$.
By construction,
the congestion in routing $\p$ in each of the edges of $X$
is caused only by the players in $\Pi$.
Since path lengths are
at most $L$,
each player in
$\Pi$ can use at most $L$ edges in $X$.
Hence, $M \leq L \cdot |\Pi|$.
Consequently,
$|X| (C - 2 \lg n) \leq L \cdot |\Pi|$,
which gives:
$C \leq (L \cdot |\Pi|)/|X| + 2 \lg n$.

By the definition of $f(X)$, in the optimal routing $\p^*$
each user in $\Pi$ has to use at least one edge in $f(X)$.
Thus, edges in the optimal routing $\p^*$,
the edges in $f(X)$ are used at least $\Pi$ times.
Thus, there is some edge $e \in f(X)$
with $C_e(\p^*) \geq |\Pi| / |f(X)|$.
Therefore, $C^* \geq |\Pi| / |f(X)|$.
Since $|f(X)| \leq 2|X|$,
we obtain $|\Pi| \leq  2 C^* \cdot|X|$.
Therefore:
$C \leq 2 L C^* + 2 \lg n.$
\end{proof}

\begin{theorem}[Price of anarchy in max games]
\label{theorem:max-anarchy}
For any max game $\R$
it holds that $PoA = O(L + \log n)$.
\end{theorem}

\begin{proof}
Suppose that $\p$ is the worst Nash-routing with maximum social cost.
We have $PoA = SC(\p) / SC(\p^*)$.
If $C \geq D + 2 \lg n + 2$,
then $SC(\p) = C$.
From Theorem \ref{theorem:max-upper-bound},
$PoA
\leq C / \max(C^*, D^*)
\leq (2 L C^* + 2\lg n) / C^*
\leq 2 L + 2 \lg n$.
If $C < D + 2 \lg n + 2$,
then $SC(\p) < D + 2 \lg n + 2$; thus
$PoA \leq L + 2\lg n + 2$.
Hence, in both cases $PoA = O(L + \log n)$.
\end{proof}

There is a max game that shows that
the result of Theorem~\ref{theorem:max-anarchy} is tight
in the worst case.
Consider a ring network with $n$ nodes and $n$ edges.
Give the same orientation to the edges,
so that each edge has one left node and one right node.
For each edge $e_i$, there is a corresponding player $i$
whose source is the left node and the destination is the right node of the edge.
The strategy set of each player has two paths: path $p_i$ which is only the edge $e_i$,
and path $p'_i$ which goes around the ring.
Note that $\p = [p_1, p_2, \ldots, p_n]$ is a Nash-routing
with social cost $1$.
However, $\p' = [p_1, p_2, \ldots, p_n]$ is also a Nash routing
with social cost $n-1$.
Thus, the price of anarchy is $O(n) = O(L)$.

\section{Sum Games}
\label{section:sum-games}

Let $R = (\N,G,\paths)$
be a routing game
such that for any routing $\p$
the social cost function is $SC(\p) = C(\p) + D(\p)$,
and the player cost function $pc_i(\p) = C_i(\p) + D_i(\p)$.
We refer to such routing games as {\em sum games}.
We first show that such games have instances without Nash-routings.
Then we describe a variation of sum games, that we call sum-bucket games,
which are stable and their equilibria have good properties.

\subsection{A Sum Game without Nash-routings}
\label{section:sum-instability}

\begin{theorem}
\label{theorem:sum-instability}
There is sum game instance $\R = (\N,G,\paths)$
that has no Nash-routing.
\end{theorem}

\begin{proof}
The graph $G$ is depicted in the figure below.
There are seven players, namely, $\N = \{1, \ldots, 7\}$.
Players 1, 2, and 3, have respective strategy sets $\paths_1 = \{p_1, p'_1\}$,
$\paths_2 = \{p_2, p'_2 \}$, and $\paths_3 = \{p_3, p'_3\}$.
In the figure
for player $i = 1, \ldots, 3$ the respective source and destination
nodes are $u_i$ and $v_i$.
There are six critical edges denoted $e_1, \ldots, e_6$
that the paths use and which are shown in the figure as straight horizontal lines.
These edges may have congestion larger than 1.
The squiggly part of the paths are assumed to have congestion 1 and their length
and their lengths are chosen so that the following relations hold:
$|p'_1| = |p_1| - 2$, $|p'_2| = |p_2| + 3$, and $|p'_3| = |p_3| + 3.$
\begin{center}
\resizebox{3.5in}{!}{\begin{picture}(0,0)%
\includegraphics{sum-unstable.pstex}%
\end{picture}%
\setlength{\unitlength}{3947sp}%
\begingroup\makeatletter\ifx\SetFigFont\undefined%
\gdef\SetFigFont#1#2#3#4#5{%
  \reset@font\fontsize{#1}{#2pt}%
  \fontfamily{#3}\fontseries{#4}\fontshape{#5}%
  \selectfont}%
\fi\endgroup%
\begin{picture}(5668,3019)(1179,-3212)
\put(1913,-1569){\makebox(0,0)[rb]{\smash{{\SetFigFont{12}{14.4}{\rmdefault}{\mddefault}{\updefault}$p_1$}}}}
\put(5960,-1552){\makebox(0,0)[lb]{\smash{{\SetFigFont{12}{14.4}{\rmdefault}{\mddefault}{\updefault}$p_1$}}}}
\put(5670,-2576){\makebox(0,0)[lb]{\smash{{\SetFigFont{12}{14.4}{\rmdefault}{\mddefault}{\updefault}$p'_1$}}}}
\put(1888,-2636){\makebox(0,0)[rb]{\smash{{\SetFigFont{12}{14.4}{\rmdefault}{\mddefault}{\updefault}$p'_1$}}}}
\put(2468,-2525){\makebox(0,0)[lb]{\smash{{\SetFigFont{12}{14.4}{\rmdefault}{\mddefault}{\updefault}$p_2$}}}}
\put(3911,-2439){\makebox(0,0)[rb]{\smash{{\SetFigFont{12}{14.4}{\rmdefault}{\mddefault}{\updefault}$p_2$}}}}
\put(3194,-2371){\makebox(0,0)[lb]{\smash{{\SetFigFont{12}{14.4}{\rmdefault}{\mddefault}{\updefault}$p_3$}}}}
\put(4552,-2192){\makebox(0,0)[lb]{\smash{{\SetFigFont{12}{14.4}{\rmdefault}{\mddefault}{\updefault}$p_3$}}}}
\put(3254,-1628){\makebox(0,0)[rb]{\smash{{\SetFigFont{12}{14.4}{\rmdefault}{\mddefault}{\updefault}$p'_3$}}}}
\put(5030,-570){\makebox(0,0)[rb]{\smash{{\SetFigFont{12}{14.4}{\rmdefault}{\mddefault}{\updefault}$p'_3$}}}}
\put(2400,-1791){\makebox(0,0)[lb]{\smash{{\SetFigFont{12}{14.4}{\rmdefault}{\mddefault}{\updefault}$p'_2$}}}}
\put(4091,-561){\makebox(0,0)[rb]{\smash{{\SetFigFont{12}{14.4}{\rmdefault}{\mddefault}{\updefault}$p'_2$}}}}
\put(4031,-937){\makebox(0,0)[b]{\smash{{\SetFigFont{12}{14.4}{\rmdefault}{\mddefault}{\updefault}$p_1 p'_3$}}}}
\put(4970,-937){\makebox(0,0)[b]{\smash{{\SetFigFont{12}{14.4}{\rmdefault}{\mddefault}{\updefault}$p_1 p_2 p_3$}}}}
\put(3075,-945){\makebox(0,0)[b]{\smash{{\SetFigFont{12}{14.4}{\rmdefault}{\mddefault}{\updefault}$p_1 p'_2$}}}}
\put(3075,-2815){\makebox(0,0)[b]{\smash{{\SetFigFont{12}{14.4}{\rmdefault}{\mddefault}{\updefault}$p'_1 p_2$}}}}
\put(4014,-2824){\makebox(0,0)[b]{\smash{{\SetFigFont{12}{14.4}{\rmdefault}{\mddefault}{\updefault}$p'_1 p_3$}}}}
\put(1546,-2064){\makebox(0,0)[rb]{\smash{{\SetFigFont{10}{12.0}{\rmdefault}{\mddefault}{\updefault}$u_1$}}}}
\put(1999,-2064){\makebox(0,0)[rb]{\smash{{\SetFigFont{10}{12.0}{\rmdefault}{\mddefault}{\updefault}$u_2$}}}}
\put(2972,-2055){\makebox(0,0)[rb]{\smash{{\SetFigFont{10}{12.0}{\rmdefault}{\mddefault}{\updefault}$u_3$}}}}
\put(2853,-1236){\makebox(0,0)[b]{\smash{{\SetFigFont{14}{16.8}{\rmdefault}{\mddefault}{\updefault}$e_1$ }}}}
\put(3775,-1244){\makebox(0,0)[b]{\smash{{\SetFigFont{14}{16.8}{\rmdefault}{\mddefault}{\updefault}$e_2$}}}}
\put(4722,-1236){\makebox(0,0)[b]{\smash{{\SetFigFont{14}{16.8}{\rmdefault}{\mddefault}{\updefault}$e_3$}}}}
\put(2836,-3131){\makebox(0,0)[b]{\smash{{\SetFigFont{14}{16.8}{\rmdefault}{\mddefault}{\updefault}$e_4$}}}}
\put(5200,-1219){\makebox(0,0)[b]{\smash{{\SetFigFont{10}{12.0}{\rmdefault}{\mddefault}{\updefault}$(+1)$}}}}
\put(6481,-2055){\makebox(0,0)[lb]{\smash{{\SetFigFont{10}{12.0}{\rmdefault}{\mddefault}{\updefault}$v_1$}}}}
\put(4731,-3131){\makebox(0,0)[b]{\smash{{\SetFigFont{14}{16.8}{\rmdefault}{\mddefault}{\updefault}$e_6$}}}}
\put(3783,-3140){\makebox(0,0)[b]{\smash{{\SetFigFont{14}{16.8}{\rmdefault}{\mddefault}{\updefault}$e_5$}}}}
\put(4961,-2824){\makebox(0,0)[b]{\smash{{\SetFigFont{12}{14.4}{\rmdefault}{\mddefault}{\updefault}$p'_1$}}}}
\put(3305,-3122){\makebox(0,0)[b]{\smash{{\SetFigFont{10}{12.0}{\rmdefault}{\mddefault}{\updefault}$(+3)$}}}}
\put(4253,-3122){\makebox(0,0)[b]{\smash{{\SetFigFont{10}{12.0}{\rmdefault}{\mddefault}{\updefault}$(+3)$}}}}
\put(5201,-3105){\makebox(0,0)[b]{\smash{{\SetFigFont{10}{12.0}{\rmdefault}{\mddefault}{\updefault}$(+4)$}}}}
\put(6459,-643){\makebox(0,0)[lb]{\smash{{\SetFigFont{10}{12.0}{\rmdefault}{\mddefault}{\updefault}$v_2$}}}}
\put(5985,-651){\makebox(0,0)[lb]{\smash{{\SetFigFont{10}{12.0}{\rmdefault}{\mddefault}{\updefault}$v_3$}}}}
\put(5931,-1101){\makebox(0,0)[lb]{\smash{{\SetFigFont{12}{14.4}{\rmdefault}{\mddefault}{\updefault}$p_2$}}}}
\put(5628,-737){\makebox(0,0)[rb]{\smash{{\SetFigFont{12}{14.4}{\rmdefault}{\mddefault}{\updefault}$p_3$}}}}
\end{picture}%
}
\end{center}

Players 4 to 7 are ``passive'' in the sense that they
have only one path in their strategy sets, and their
sole purpose is to create additional congestion on edges $e_3, e_4, e_5, e_6$
(the paths of these players are not shown explicitly in the figure).
In particular, the passive players cause additional congestion $1$ to edge $e_3$,
additional congestion $3$ to $e_4$ and $e_5$, and additional congestion $4$ to $e_6$.
The additional congestion
is depicted in the figure inside a parenthesis under each edge.

Since the only ``active'' players are $1$, $2$, and $3$,
and each player has two path choices,
there are eight possible different routings.
We examine each routing and prove that it is not a Nash-routing.
We use the vector $[p_1,p_2,p_3]$ to denote a routing where
the $i$th position of the vector contains the path choice of user $i$.
By setting explicit values to the path lengths,
and computing the player costs in each routing,
we find that:
player 1 is not locally optimal in routings
$[p_1, p_2,p_3]$ and $[p'_1, p'_2,p'_3]$;
player 2 is not locally optimal in routings
$[p'_1, p_2,p_3]$, $[p_1, p'_2,p'_3]$, $[p_1, p'_2,p_3]$, and $[p'_1, p_2,p'_3]$;
and player 3 is not locally optimal in routings
$[p'_1, p'_2,p_3]$ and $[p_1, p_2,p'_3]$.
\end{proof}

\subsection{Sum-bucket Games}

Here we describe {\em sum-bucket games},
which are variation of sum games that are stable and their equilibria have good properties.
Let $\R = (\N,G,\paths)$ denote a sum-bucket routing game.
The paths in $\paths$ are divided into {\em buckets}
$B_0, \ldots, B_{\lg L}$ so that
$B_k$ is a set that contains all paths
whose lengths are in range $[2^{k}, 2^{k+1})$.
(We use $\lg L$ instead of $\lceil \lg L \rceil$ to avoid notational clutter.)
For any path $p \in \paths$, let $B(p)$ denote the index of the
bucket that $p$ belongs to; namely, if $p \in B_k$, then $B(p) = k$.

Consider now a routing $\p$.
We define the {\em normalized length} of path $p \in \p$
as $\D_p(\p) = 2^{B(p)+1} - 1$
(which is the maximum possible path length in bucket $B(p)$).
Note that all the paths in the same bucket have the same normalized length.
For any path $p$ we define the {\em normalized congestion} $\C_{p}(\p)$ to be the
maximum congestion on any of the edges of path $p$
which is caused by the paths of routing $\p$ which belong to bucket $B(p)$.

Player $i$'s cost is $pc_i(\p) = \C_i(\p) + \D_i(\p)$.
The player is allowed to switch from one bucket to another.
The normalized congestion of $\p$ is $\C(\p) = \max_{p_i \in \p} \C_i(\p)$,
and the normalized length of $\p$ is $\D(\p) = \max_{p_i \in \p} \D_i(\p)$.
The social cost of game $\R$ is $SC(\p) = \C(\p) + \D(\p)$.
Below, we show that sum-bucket games are stable
and then we bound the price of anarchy.

\subsubsection{Existence of Nash-routings in Sum-bucket Games}
\label{section:sum-bucket-stability}

Here we show that sum-bucket games have Nash-routings.
We use the same technique as in Section~\ref{section:max-stability},
where we order the routings and prove that greedy moves 
give smaller order routings.
Let $\R = (\N,G,\paths)$
be a sum-bucket routing game.
Let $r = N + 2L-1$ (this is the maximum possible player cost).
For any routing $\p$
we define
the \emph{routing vector}
$M(\p)=[m_1(\p), \ldots, m_r(\p)]$
such that $m_j(\p)$ is the number of paths in $\p$ with cost $j$.
We define a total order on the routings,
with respect to their vectors,
in exactly the same way as wee did for the max games in Section \ref{section:max-stability}.
Using similar techniques as in the max games, we can prove that
if a greedy move by player $i$ takes
a routing $\p$ to
a new routing $\p'$, then
$\p' < \p$.
Since there are only a finite number of routings,
every best response dynamic converges in a finite time
to a Nash-routing.
Therefore, we get:

\begin{theorem}[Stability of sum-bucket games]
\label{theorem:sum-bucket-stability}
For any sum-bucket game $\R$,
every best response dynamic converges to a Nash-routing.
\end{theorem}

\subsubsection{Price of Anarchy in Sum-bucket Games}
\label{section:sum-bucket-anarchy}

From Theorem \ref{theorem:sum-bucket-stability},
every sum-bucket game has at least one Nash-routing.
Here, we bound the price of anarchy.
Consider a sum-bucket routing game $\R = (\N,G,\paths)$,
where $G$ has $n$ nodes.
Consider a Nash-routing $\p$.
Denote $\C = \C(\p)$ and $\D = \D(\p)$.
Let $\p^*$ be the optimum (coordinated) routing
with minimum social cost.
Denote $\C^* = \C(\p^*)$
and $\D^* = \D(\p^*)$.
Note that each payer $i \in \N$
has a path in $p_i \in \p$ and a corresponding ``optimal'' path $p^*_i \in \p^*$
from the player's source to the destination node.
For any player $i$, let $s_i$ denote the shortest path in $\paths$ which connects
the source and destinations nodes of $i$.
We now relate the paths lengths with the congestion:

\begin{lemma}
\label{theorem:length-congestion}
In Nash-routing $\p$, for any player $i$ with $\C_i \geq \C - x$,
where $x \geq 0$,
it holds that $|p_i| \leq |s_i| + x + 1$.
\end{lemma}

\begin{proof}
Suppose for the sake of contradiction that there is a player $i$ with $|p_i| > |s_i| + x + 1$.
Then,
$pc_i = \C_i + \D_i \geq \C - x + \D_i
\geq \C - x + |p_i| > \C - x + |s_i| + x + 1= \C + |s_i| + 1.$
Clearly, $|s_i| \leq 2^{B(s_i) + 1} - 1$.
If user $i$ was to switch to path $s_i$
its cost would be $pc'_i \leq \C + 1 + 2^{B(s_i) + 1} - 1$,
since the normalized length of $s_i$ is $2^{B(s_i) + 1} - 1$,
and $s_i$ has normalized congestion at most $\C$
before player $i$ switches,
and the normalized congestion of path $s_i$
increases to at most $\C+1$ after player $i$ switches to it.
Therefore,
$pc'_i \leq \C + 1 + |s_i| < pc_i$.
Thus, in $\p$ player $i$ would not be optimal,
which is a contradiction, since $\p$ is a Nash-routing.
Therefore, $|p_i| \leq |s_i| + x + 1$, as needed.
\end{proof}

For each edge $e \in G$ let $\Pi_e(\p)$ denote the set of players
whose paths in routing $\p$ use edge $e$.
Let $\C_{e,p_i}(\p)$ denote the number of packets that use edge $e$ in $\p$
and are in the same bucket as $p_i$.
Let $\C_e(\p) = \max_{p_i} C_{e,p_i}(\p)$
denote the normalized congestion of edge $e$.
For any edge $e$, we define $f(e,i)$ to be a set that contains all edges $e' \in p^*_i$
with $\C_{e',p^*_i}(\p) \geq \C_{e,p_i}(\p) - \D^*$.
Let $f(e) = \cup_{i \in \Pi_e(\p)} f(e,i)$,
and for any set of edges $X$,
$f(X) = \bigcup_{e \in X} f(e)$.
It can be shown that in Nash-routing $\p$ it holds $|f(e,i)| \geq 1$.

\begin{lemma}
\label{theorem:sum-bucket-expansion}
Let $Z$ be the set that contains all
edges $e$ with congestion $\C_e(\p) \geq \C - 2 \D^* \cdot \lg n$.
There is a set of edges $X \subseteq Z$
with $|f(X)| \leq 2|X|$.
\end{lemma}

\begin{proof}
We recursively construct a sets of edges $E_0, \ldots, E_{2\lg n}$,
such that $E_i = E_{i-1} \cup f(E_{i-1})$,
and set $E_0$ contains
all the edges $e$ with congestion $C_e(\p) = \C$.
We can show that there is a $j$, $0 \leq j \leq 2\lg n$,
such that $|f(E_j)| \leq 2 |E_j|$.
Suppose for contradiction that such a $j$ does not exist.
Thus for all $j$, $0 \leq j \leq 2\lg n$,
it holds that $|f(E_j)| > 2|E_j|$.
In this case, it it straightforward to show that $|E_k| > 2|E_{k-1}|$,
for any $1 \leq k \leq 2\lg n$.
Since $|E_0| \geq 1$,
it holds that $|E_{2\lg n}| > 2^{2\lg n} = n^2$.
However, this is a contradiction, since
the number of edges in $G$ cannot exceed $n^2$.
Thus, there is a $j$ with $|f(E_j)| > 2|E_j|$.
We will set $X = E_j$.

It only remains to show that $X \subseteq Z$.
It suffices to show that for any $E_k$ and $e \in E_k$, $\C_{e}(\p) \geq \C-k \D^*$,
where $0 \leq k \leq 2\lg n$.
We prove this by induction on $k$.
For $k = 0$ we have that every edge in $e$ has $\C_e(\p) = \C = \C - 0 \cdot \D^*$,
thus the claim trivially holds.
For the induction hypothesis, suppose that the claim holds for any $k = t < 2 \log n$.
In the induction step we will prove that the claim holds also for $k = t+1$.
We have that $E_{t+1} = E_{t} \cup f(E_{t})$.
By induction hypothesis,
for any $e \in E_t$, $\C_{e}(\p) \geq \C - t \D^*$
(thus, for any $e \in E_t$ it holds that  $\C_{e}(\p) \geq \C - (t+1)D^*)$).
Thus, for any $e \in E_k$ there is at least one path $p_i$
with $\C_{e,p_i}(\p) \geq \C - t\D^*$.
By definition of $f(e,i)$,
every edge $e \in f(e,i)$ has the property that
$\C_{e',p^*_i}(\p) \geq \C_{e,p_i}(\p) - \D^* \geq \C - (t+1)\D^*$.
Thus, there is at least one path $p' \in \p$
which is in the same bucket with $p^*_i$ (namely, $B(p') = B(p^*_i)$)
for which it holds that $\C_{e',p'}(\p) = \C_{e',p^*_i}(\p)$.
Therefore, $\C_{e'}(\p) \geq  \C - (t+1)\D^*$.
Consequently, from the definition of $f(E_{t})$
it follows that for any edge $e' \in f(E_{t})$ it holds that $\C_{e'}(\p) \geq  \C - (t+1)\D^*$.
By considering the union of $E_{t} \cup f(E_{t})$,
we have that the claim holds for any edge in $e \in E_{t+1}$, as needed.
\end{proof}

\begin{lemma}
\label{theorem:sum-bucket-upper-bound}
In Nash-routing $\p$ it holds that $\C \leq 18 \C^* \cdot \D^* \cdot \lg^2 n$.
\end{lemma}

\begin{proof}
From Lemma \ref{theorem:sum-bucket-expansion},
there is a set of edges $X \subset Z$ with $|f(X)| \leq 2|X|$.
For each $e \in Z$ it holds that $\C_e(\p) \geq \C - 2 \D^* \cdot \lg n$.
Let $\Pi = \bigcup_{e \in X} \Pi_e(\p)$,
that is, $\Pi$ is the set of players which in routing $\p$
their paths use edges in $X$.
Let $M = \sum_{e \in X} \C_e(\p)$,
which denotes the total ``normalized'' utilization of the
edges in $X$.
We have that $M \geq |X| (\C - 2 \D^* \cdot \lg n)$.

By construction,
the congestion in routing $\p$ in each of the edges of $X$
is caused only by the players in $\Pi$.
We can bound the path lengths of the players in $\Pi$ with respect to $\D^*$ as follows.
Consider a player $i \in \Pi$.
We have that $C_i(\p) \geq \C - 2 \bar D^* \cdot \lg n$.
From Lemma \ref{theorem:length-congestion} and the fact that $|s_i| \leq D^*$,
we obtain:
$|p_i| \leq |s_i| + 2 \D^* \cdot \lg n + 1
\leq \D^* + 2 \D^* \cdot \lg n + 1
\leq 4 \D^* \cdot \lg n.$
Thus, the path length of every player in $\Pi$ is at most $K = 4 \D^* \cdot \lg n$.

$M$ can also be bounded as $M \leq K |\Pi|$.
Consequently,
$|X| (\C - 2 \D^* \cdot \lg n) \leq  K |\Pi|$,
which gives:
\begin{equation}
\label{eqn:1}
\C \leq \frac {K |\Pi|} {|X|} + 2 \D^* \cdot \lg n
= \frac {4 \D^* \cdot \lg n \cdot |\Pi|} {|X|} + 2 \D^* \cdot \lg n.
\end{equation}
Since $|f(e,i)| \geq 1$,
in the optimal routing $\p^*$
the path of each user in $\Pi$ has to use at least one edge in $f(X)$.
Thus, in the optimal routing $\p^*$,
the edges in $f(X)$ are used at least $|\Pi|$ times.
Thus, there is some edge $e \in f(X)$
which in $\p^*$ is used by at least $|\Pi| / |f(X)|$ paths.
Since there are $\lg L + 1$ buckets,
the normalized congestion of $e$ in one of those buckets is at least
$|\Pi| / (|f(X)| \cdot (\lg L + 1))$.
Therefore, $\C^* \geq |\Pi| / (|f(X)| \cdot (\lg L + 1))$.
Since $|f(X)| \leq 2|X|$,
we obtain:
\begin{equation}
\label{eqn:2}
|\Pi| \leq  2 \C^* \cdot|X| \cdot (\lg L + 1) \leq 4 \C^* \cdot|X| \cdot \lg n.
\end{equation}
By Combining Equations \ref{eqn:1} and \ref{eqn:2}, we get:
$ \C \leq 16 \C^* \cdot \D^* \cdot \lg^2 n + 2 \D^* \lg n
\leq 18 \C^* \cdot \D^* \cdot \lg^2 n.$
\end{proof}

When $\C > \D/4$, using Lemma \ref{theorem:sum-bucket-upper-bound}
it is straightforward to prove that $PoA = O( {(\C^* \cdot \D^* \cdot \lg^2 n )} / {(\C^* + \D^*)})$.
If $\C \leq \D/4$,
then using Lemma \ref{theorem:length-congestion},
we can prove that $PoA = O(1)$ (the details can be found in the appendix).
Therefore, we obtain the main result: 

\begin{theorem}[Price of anarchy in sum-bucket games]
\label{theorem:sum-bucket-anarchy}
For any sum-bucket game $\R$
it holds:
$$PoA = O\left(\frac {\C^* \cdot \D^* } {\C^* + \D^*} \cdot \lg^2 n\right).$$
\end{theorem}

There is a sum-bucket game $\R = (\N, G, \paths)$
that shows that the result of Theorem \ref{theorem:sum-bucket-anarchy}
is tight in non-trivial cases.
All the players have the same source node $u$ and destination node $v$.
Let $a = \sqrt N$ (suppose for simplicity that $a = \sqrt N = \lceil \sqrt N \rceil$).
There is a path $p$ of length $a$ from $u$ to $v$.
There are $a$ edge-disjoint paths $Q = \{ q_1, \ldots, q_a \}$ from
$u$ to $v$ so that each path $q_i$ has length $a$ and uses one edge of $p$.
Each player has two paths in her strategy set: one is path $p$ and
the other is a path from $Q$.
Further, for each path $q_i$ there are $a$ players that have $q_i$ in their strategy sets.
Let $\p$ be the routing where every player chooses path $p$.
Then, $\p$ is a Nash-routing with social cost $a + N$.
Let $\p'$ be the routing where every player chooses the alternative path in $Q$.
Then, $\p'$ is also a Nash-routing with social cost $2a$.
Thus, $PoA \geq SC(\p') / SC(\p) \geq (a+N) / (a+1) = \Omega(\sqrt N) = \Omega(\sqrt n)$.
Routing $\p'$ is an optimum routing with the smallest social cost $\C^* = \C(\p) = a$
and $\D^* = D(\p) = a$.
Thus, from Theorem \ref{theorem:sum-bucket-anarchy},
$PoA = O((\C^* \cdot \D^* \cdot \lg^2 n)/ (\C^* + \D^*)) = O(a \cdot \lg^2 n) = O(\sqrt n \cdot \log^2 n)$.
Hence, the price of anarchy has to be
within a $\log^2 n$ factor from the bound provided in Theorem \ref{theorem:sum-bucket-anarchy}.

\section{Conclusions}
\label{section:conclusions}

In this work we provided the first study (to our knowledge) of bicriteria routing games,
where the players attempt to simultaneously optimize two parameters:
their path congestion and length.
The motivation is the existence of efficient 
packet scheduling algorithms which deliver the packets 
in time proportional to the social cost.
We examined max games  
and sum games.
Max games stabilize, but their price of anarchy is high.
Sum games do not stabilize, but they can give better price of anarchy.
We then give the approximate sum-bucket games which always stabilize 
and preserve the good properties of sum games.
Surprisingly, arbitrary sum-bucket game equilibria provide good
approximations to the original coordinated routing problem.

Several open problems remain to examine. We studied two particular functions
of the bicriteria, namely, the max and the sum functions. There are other functions, 
for example a weighted sum, that could provide similar or better results.
It would also be interesting to add additional parameters.
The original $C+D$ sum games do not stabilize in general.
However, there exist interesting instances which stabilize. 
For example, it can be easily shown that the games where the available paths have equal
lengths always stabilize. 
It would be interesting to find a general 
characterization of the game instances that stabilize. 
Another interesting problem is to provide time efficient 
algorithms for finding equilibria in our games.

\newpage
{
\bibliographystyle{plain}
\bibliography{mypapers,masterbib,routing,oblivious,game}
}

\newpage
\pagenumbering{roman}
\appendix
\section{Additional Proofs of Section \ref{section:max-stability}}

\begin{lemma}
\label{theorem:max-minimum-optimal}
Minimum routing $\pm$ achieves optimal social cost,
that is, $SC(\pm) \leq SC(\p)$,
for any other routing $\p \neq \pm$.
\end{lemma}

\begin{proof}
Suppose for contradiction that there exists a routing $\p \geq \pm$
with $SC(\p) < SC(\pm)$.
Let $SC(\p) = \max(C(\p),D(\p)) = k_1$,
and $SC(\pm) = \max(C(\pm), D(\pm)) = k_2$.
Clearly, $k_1 < k_2$.
Therefore,
in the vector $M(\p) = [m_1,\ldots,m_r]$ it holds that
$m_{k_1} \neq 0$ and $m_k = 0$ for $k > k_1$.
Similarly,
in the vector $M(\pm) = [{\widehat m}_1,\ldots, {\widehat m}_r]$
it holds that ${\widehat m}_{k_2} \neq 0$ and ${\widehat m}_k = 0$ for $k > k_2$.
Therefore, $M(\pm) > M(\p)$,
contradicting the fact that $\p \geq \pm$.
\end{proof}

\section{Additional Proofs of Section \ref{section:sum-instability}}

In the table below
we calculate the congestions, path lengths, and player costs for
each routing scenario of the game in Theorem \ref{theorem:sum-instability}.
We set the specific lengths as:
$|p_1| = 10$, $|p'_1| = 8$, $|p_2| = 7$, $|p'_2| = 10$, $|p_3|  = 7$, and $|p'_3| = 10$.
For a player $i \in \{ 1, 2, 3 \}$ and routing $\p$ we define the {\em complementary routing}
to be the one where player $i$ chooses the alternative path.
For example, for player $2$ the complementary
routing of $[p_1,p_2,p_3]$ is $[p_1, p'_2, p_3]$.
A player is locally optimal in a routing if
the complementary routing does not give a lower cost for the player.
Using the table it is easy to determine whether a player
is locally optimal or not by examining the respective costs in the
complementary routings.
In this way, we find the non-locally players which are shown
in the rightmost column of the table.
\begin{center}
\begin{tabular}{|l||l|l|l||l|l|l||l|l|l||c|}
  \hline
routing            & $C_1$ & $D_1$ & $pc_1$ & $C_2$ & $D_2$ & $pc_2$ & $C_3$  & $D_3$ & $pc_3$ & {\footnotesize Non locally optimal players}\\
\hline \hline
$[p_1,p_2,p_3]$    &   4   &  10   &   14   &  4    &   7   &  11    &   4    &   7   &  11     & player 1      \\ \hline
$[p'_1,p_2,p_3]$   &   5   &   8   &   13   &  5    &   7   &  12    &   5    &   7   &  12     & player 2      \\ \hline
$[p'_1,p'_2,p_3]$  &   5   &   8   &   13   &  1    &  10   &  11    &   5    &   7   &  12     & player 3      \\ \hline
$[p'_1,p'_2,p'_3]$ &   5   &   8   &   13   &  1    &  10   &  11    &   1    &  10   &  11     & player 1      \\ \hline
$[p_1, p'_2,p'_3]$ &   2   &  10   &   12   &  2    &  10   &  12    &   2    &  10   &  12     & player 2      \\ \hline
$[p_1, p_2, p'_3]$ &   3   &  10   &   13   &  4    &   7   &  11    &   2    &  10   &  12     & player 3      \\ \hline
$[p_1, p'_2, p_3]$ &   3   &  10   &   13   &  2    &  10   &  12    &   4    &   7   &  11     & player 2      \\ \hline
$[p'_1, p_2, p'_3]$&   5   &   8   &   13   &  5    &   7   &  12    &   1    &  10   &  11     & player 2      \\ \hline
\end{tabular}
\end{center}

\section{Additional Proofs of Section \ref{section:sum-bucket-stability}}

\begin{lemma}
\label{lemma:sum-bucket-greedy}
If a greedy move by player $i$ takes
a routing $\p$ to
a new routing $\p'$, then
$\p' < \p$.
\end{lemma}

\begin{proof}
Let $p_i$ and $p'_i$ denote the paths of $i$
in routings $\p$ and $\p'$, respectively.
Let $pc_i(\p) = \C_i(\p) + \D_i(\p) = z_1$
and
$pc_i(\p') = \C_i(\p') + \D_i(\p') = z_2$.
Since player $i$ decreases its cost in $\p'$, $z_2 < z_1$.
Consider now the vectors of the routings
$M(\p) = [m_1,\ldots,m_r]$
and $M(\p') = [m'_1,\ldots,m'_r]$.
We will show that $M(\p) < M(\p')$.

Let $B(p_i) = k_1$ and $B(p'_i) = k_2$.
Let $V$ denote the set of players whose cost increases in $\p'$
with respect to their cost in $\p$.
Next, we show that for any $j \in V$ it holds that $pc_j(\p') \leq c_i(\p')$,
which will help us to prove the desired result.

Let $p_j$ be the path of $j \in V$.
Let $B(p_j) = k_3$
(note that $j$ does not switch paths and buckets between $\p$ and $\p'$).
If $k_3 \neq k_1$ and $k_3 \neq k_2$ then the cost of $j$ remains
unaffected between $\p$ and $\p'$.
Thus, either $k_3 = k_1$ or $k_3 = k_2$.
If $k_3 = k_1$, then the cost of $j$ can only decrease
from $\p$ to $\p'$, since path $p_i$ can no longer affect path $p_j$.
Therefore, it has to be that $k_3 = k_2$.
Suppose, for the sake of contradiction, that $pc_j(\p') > pc_i(\p')$.
Then, $\C_j(\p') + \D_j(\p') > \C_i(\p') + \D_i(\p')$.
Since both paths are in the same bucket,
$\D_j(\p') = \D_i(\p')$, which implies that $\C_j(\p') > \C_i(\p')$.
Since $j \in V$, $\C_j(\p) < \C_j(\p')$.
The increase in normalized congestion of $p_j$ in $\p'$
can only be caused by $p'_i$ because it shares an edge with $p_j$
with congestion $\C_j(\p')$.
However, this is impossible since it would imply that $\C_i(\p') \geq \C_j(\p')$.
Therefore, $pc_j(\p') \leq pc_i(\p')$.

Consequently, in vector $\p'$ all the entries in positions
$j_2+1, \ldots, r$ do not increase with respect to $\p$.
Further, because $i$ switches paths, $m_{j_1} > m'_{j_2}$.
Thus, $M(\p) < M(\p')$, as needed.
\end{proof}

\section{Additional Proofs of Section \ref{section:sum-bucket-anarchy}}

\begin{lemma}
\label{theorem:f-at-least-one}
In Nash-routing $\p$, for every player $i$ and edge $e$ it holds that $|f(e,i)| \geq 1$.
\end{lemma}

\begin{proof}
Suppose that $|f(e,i)| = 0$.
Then for every edge $e' \in p^*_i$ it holds that $\C_{e',p^*_i}(\p) < \C_{e,p_i}(\p) - \D^*$.
Let $\C' = \max_{e' \in p^*_i} \C_{e',p^*_i}(\p)$.
If player $i$ was to choose path $p^*_i$ its cost would be:
$pc'_i \leq \C' + 1 + \D_{p^*_i} < \C_{e,p_i}(\p) - \D^* + 1 + \D^*
= \C_{e,p_i}(\p) + 1 \leq \C_i(\p) + \D_i(\p) = pc_i(\p).$
Thus, $pc'_i < pc_i(\p)$ which implies that player $i$ is not locally optimal in Nash-routing $\p$,
a contradiction.
\end{proof}

\paragraph{Proof of Theorem~\ref{theorem:sum-bucket-anarchy}:}
\hfill \break
%
%
Suppose that $\p$ is the worst Nash-routing
with maximum social cost.
We have that
$
PoA \leq {SC(\p)} / {SC(\p^*)} \leq {(\C + \D)} / {(\C^* + \D^*)}.
$
We examine two cases:
\begin{itemize}
\item
$\C > \D/4$:
In this case $\D = O(\C)$.
From Lemma \ref{theorem:sum-bucket-upper-bound},
$\C \leq 18 \C^* \cdot \D^* \cdot \lg^2 n$.
Therefore:
$
PoA
= O( {\C } / {(\C^* + \D^*)})
= O( {(\C^* \cdot \D^* \cdot \lg^2 n )} / {(\C^* + \D^*)}).
$

\item
$\C \leq \D/4$:
Let $p_i \in \p$ be the path with maximum cost in $\p$.
Clearly, $pc_i(\p) = \C_i + \D_i \geq \D$.
Further,
$0 \leq \C_i(\p) \leq \C \leq \D / 4$.
Thus, $\D_i \geq \D - \C_i \geq \D - \D / 4 = 3 \D /4$.
Since $\C_i \geq 0 = \C - \C$,
Lemma \ref{theorem:length-congestion} gives
$|p_i| \leq |s_i| + \C + 1 \leq \D^* + \D / 4 + 1$.
We have that $\D_i /2  < |p_i|$.
Therefore,
$\D_i /2 < \D^* + \D / 4 + 1$,
which gives:
$3 \D / 8 < \D^* + \D / 4 + 1$.
Thus, $\D < 8( \D^* + 1 ) \leq 16 \D^*$.
In order words, $\D = O(\D^*)$.
Since, $\C = \D$, we obtain:
$PoA
= O( {\D } / {(\C^* + \D^*)} )
= O( {\D^*} / {(\C^* + \D^*)}) = O(1).
$
\end{itemize}
By combining the two above cases we obtain the desirable result.
\qed

\end{document}